\documentclass[letterpaper]{scrartcl}

\usepackage{group}
\usepackage{placeins}
\usepackage{pgfplots}
\usepackage{caption}
\usepackage{subcaption}
\pgfplotsset{width=8cm,compat=1.9}

\DeclareSymbolFont{largesymbolsA}{U}{txexa}{m}{n}

\renewcommand{\cond}{\ensuremath{\mathit{COND}}\xspace}
\newcommand{\cnl}{\ensuremath{\mathit{CNL}}\xspace}

\DeclareMathOperator*{\E}{\mathbb{E}}

\allowdisplaybreaks

\begin{document}

\title{Robust Bounds on Choosing from \\Large Tournaments}

\author{
Christian Saile\\
Technical University of Munich\\Germany
\and
Warut Suksompong\\
University of Oxford\\United Kingdom
}
  
\date{}
  
\maketitle

\begin{abstract}
Tournament solutions provide methods for selecting the ``best'' alternatives from a tournament and have found applications in a wide range of areas. Previous work has shown that several well-known tournament solutions almost never rule out any alternative in large random tournaments. Nevertheless, all analytical results thus far have assumed a rigid probabilistic model, in which either a tournament is chosen uniformly at random, or there is a linear order of alternatives and the orientation of all edges in the tournament is chosen with the same probabilities according to the linear order. In this work, we consider a significantly more general model where the orientation of different edges can be chosen with different probabilities. We show that a number of common tournament solutions, including the top cycle and the uncovered set, are still unlikely to rule out any alternative under this model. This corresponds to natural graph-theoretic conditions such as irreducibility of the tournament. In addition, we provide tight asymptotic bounds on the boundary of the probability range for which the tournament solutions select all alternatives with high probability.
\end{abstract}

\section{Introduction}

Tournaments play an important role in numerous situations as a means of representing entities and a dominance relationship between them. For instance, both the outcome of a round-robin sports competition and the majority relation of voters in an election can be represented by a tournament. A question that occurs frequently is therefore the following: Given a tournament, how can we choose the ``best'' alternatives in a consistent manner? This question has been addressed by a rich and beautiful literature on tournament solutions, which have found applications in areas ranging from sports competitions \citep{Usha76a} to multi-criteria decision analysis \citep{ArRa86a,Bouy04a} to biology \citep{Schj22a,Land53a,Slat61a,AlLe11a}. Over the past half century several tournament solutions have been proposed, two of the oldest and best-known of which are the \emph{top cycle} \citep{Good71a,Schw72a,Mill77a} and the \emph{uncovered set} \citep{Mill80a}.\footnote{For a thorough treatment of tournament solutions, we refer the reader to excellent surveys by \cite{Lasl97a} and \cite{BBH15a}.} 

Given that the purpose of tournament solutions is to discriminate the ``best'' alternatives from the remaining ones, it perhaps comes as a surprise that many common tournament solutions---including the top cycle, the uncovered set, the Banks set, and the minimal covering set---select all alternatives with high probability in a large random tournament \citep{Fey08a,ScFe11a}. Put differently, the aforementioned tournament solutions almost never exclude any alternative in a tournament chosen at random. Nevertheless, these results are based on the \emph{uniform random model}, in which all tournaments are drawn with equal probability, or equivalently each edge is oriented in one direction or the other with equal probability independently of other edges. For a large majority of applications of tournaments, one would not expect that this assumption holds. Indeed, stronger teams are likely to beat weaker teams in a sports competition, and candidates with a large base of support have a higher chance of winning an election. Moreover, real-world tournaments often exhibit a certain degree of transitivity: If alternatives $a$, $b$, and $c$ are such that $a$ dominates $b$ and $b$ dominates $c$, then it is more likely that $a$ dominates $c$ than the other way around.

A more general model of random tournaments is the \emph{Condorcet random model}, previously considered by \cite{Fran68a}, \cite{LRG96a}, \cite{Will10a} and \cite{KSV17a}. In this model, there is a linear order of alternatives, which can be interpreted as an ordering of the alternatives from strongest to weakest. For each pair of alternatives, the probability that the edge is oriented from the alternative that occurs later in the linear order to the alternative that occurs earlier in the linear order is $p$, independently of other pairs of alternatives.\footnote{By symmetry, we may assume without loss of generality that $p\leq 1/2$.} Crucially, the value of $p$ is the same for all pairs of alternatives. The Condorcet random model generalizes the uniform random model, since the latter can be obtained from the former by taking $p=1/2$. \cite{LRG96a} showed that under the Condorcet random model, the top cycle selects all alternatives as long as $p\in\omega(1/n)$. The same authors show furthermore that this bound is tight, that is, the statement no longer holds if $p\in O(1/n)$.\footnote{See, e.g., \cite{CLRS09a} for the definitions of asymptotic notations.}

Although the Condorcet random model addresses the issues raised above with regard to the uniform random model, it is still rather unrealistic for two important reasons. Firstly, in tournaments in the real world, the orientation of different edges are typically determined by different probabilities. For instance, in a sports tournament the probability that a very strong team beats a very weak team is usually higher than the probability that a moderately strong team beats a moderately weak team; a similar phenomenon can be observed in elections. Secondly, even though one can roughly order the alternatives in a tournament according to their strength, it is often the case that not all probabilities of the orientation of the edges respect the ordering. Indeed, this precisely corresponds to the notion of ``bogey teams''---weak teams that nevertheless frequently beat certain supposedly stronger teams. Given the limitations of the uniform random model and the Condorcet random model, it is natural to ask whether previous results continue to hold under more general and realistic models of random tournaments, or whether they break down as soon as we move beyond these restricted models.

In this paper, we show that a number of tournament solutions, including the top cycle and the uncovered set, still choose all alternatives with high probability under a significantly more general model of random tournaments. Unlike the Condorcet random model, our model does not rely on an ordering of the alternatives. Instead, the orientation of each edge is determined by probabilities within the range $[p,1-p]$ for some parameter $p$, and these probabilities are allowed to vary across edges. The only substantive assumption that we make is that the orientations of different edges are chosen independently from one another. Under this model, which is more general than both the uniform random model and the Condorcet random model, we establish in Section~\ref{sec:TC} that the top cycle almost never rules out any alternative as long as $p\in\omega(1/n)$, thus generalizing the result by \cite{LRG96a}. We also show that our bound is asymptotically tight, and that analogous results hold for two other tournament solutions based on the set of Condorcet winners and losers as well. Moreover, we prove in Section~\ref{sec:UC} that the uncovered set is likely to include the whole set of alternatives when $p\in\omega(\sqrt{\log n/n})$. This bound is again asymptotically tight, and the same holds for another tournament solution based on the uncovered set. Since the condition that the top cycle or the uncovered set chooses all alternatives have meaningful graph-theoretic interpretations---the top cycle is the whole set of alternatives if and only if the tournament is strongly connected,\footnote{A strongly connected tournament is also said to be \emph{strong}. Strong connectedness is equivalent to \emph{irreducibility} and to the property of having a Hamiltonian cycle \citep{Moon68a}.} and the uncovered set fails to exclude any alternative exactly when all alternatives are \emph{kings}\footnote{\label{footnote:kings}A \emph{king} is an alternative that can reach any other alternative via a directed path of length at most two \citep{Maur80a}. Therefore, all alternatives of a tournament are kings if and only if every pair of alternatives can reach each other via a directed path of length at most two. Such a tournament has been studied in graph theory and called an \emph{all-kings tournament} \citep{Reid82a}.}---we believe that our results are of independent interest in graph theory and discrete mathematics. Furthermore, the generality of our model allows us to derive consequences in Section~\ref{sec:extensions} for a different model in which tournaments are generated from random voter preferences, and we complement our theoretical results with experimental data in Section~\ref{sec:experiments}. 

\subsection{Related Work}

The study of the behavior of tournament solutions in large random tournaments goes back to \cite{MoMo62a}, who showed that the top cycle almost never rules out any alternative in a large tournament chosen uniformly at random. In fact, they proved a stronger statement that the probability that the top cycle excludes at least one alternative is inverse exponential in the number of alternatives; the estimate was later made more precise by \cite{Moon68a} in his seminal book on tournaments. \cite{Bell81a} also considered the top cycle but assumed that tournaments are generated from the preferences of a large number of voters, each with a uniform random ranking over the alternatives; he likewise found that the top cycle selects all alternatives with high probability under this assumption. \cite{Fey08a} and later \cite{ScFe11a} established results on several tournament solutions including the uncovered set, the Banks set, the Copeland set, the minimal covering set, and the bipartisan set using the uniform random model. While the uncovered set, the Banks set, and the minimal covering set are likely to include all alternatives in a large random tournament, the same event is unlikely to occur for the Copeland set. On the other hand, the bipartisan set chooses on average half of the alternatives in a random tournament of any fixed size \citep{FiRy95a}; it is the unique most discriminating tournament solution satisfying standard properties proposed in the literature \citep{BBSS14a}.

The discriminative power of tournament solutions has also been investigated empirically by \cite{BrSe15a}. Building on the observation that the distributions of real-world tournaments are typically far from uniform, these authors examined the behavior of eleven common tournament solutions on tournaments generated according to stochastic preference models and empirical data. The stochastic models that they used include the impartial culture model, the Mallows mixtures model, and the P\'{o}lya-Eggenberger urn model. They reported that under these more realistic models, most tournament solutions are in fact much more discriminating than the analytical results for uniform random tournaments suggest.

\section{Preliminaries}
\label{sec:prelims}

A tournament $T$ consists of a set $A=\{a_1,a_2,\dots,a_n\}$ of alternatives and a dominance relation. The dominance relation is an asymmetric and connex binary relation on $A$ represented by a directed edge between each unordered pair of distinct alternatives in $A$. We say that alternative $a_i$ \emph{dominates} another alternative $a_j$ if there is an edge from $a_i$ to $a_j$. An alternative is said to be a \emph{Condorcet winner} if it dominates all of the remaining alternatives, and a \emph{Condorcet loser} if it is dominated by all of the remaining alternatives. We extend the dominance relation to sets and say that a set $A'\subseteq A$ of alternatives dominates another set $A''\subseteq A$ of alternatives disjoint from $A'$ if for all $a'\in A'$ and $a''\in A''$, $a'$ dominates $a''$. A tournament is commonly interpreted as the outcome of a round-robin sports competition and as the majority relation of an odd number of voters with linear preferences. In the former interpretation, alternative $a_i$ dominating alternative $a_j$ means that the player or team represented by $a_i$ beats the player or team represented by $a_j$ in the competition. In the latter interpretation, the same dominance relation signifies that more than half of the voters prefer $a_i$ to $a_j$.

We are interested in tournament solutions, which are functions that map each tournament to a nonempty subset of its alternatives, usually referred to as the \emph{choice set}. Two simple tournament solutions are \cond, which chooses a Condorcet winner if one exists and chooses all alternatives otherwise,\footnote{Note that the set of Condorcet winners is not a tournament solution because it can be empty.} and the set of \emph{Condorcet non-losers} (\cnl), which consists of all alternatives that are not Condorcet losers. Other tournament solutions considered in this paper are the following:
\begin{itemize}
\item The \emph{top cycle} ($\tc$) is the (unique) smallest set of alternatives such that all alternatives in the set dominate all alternatives not in the set;
\item The \emph{uncovered set} (\uc) consists of all alternatives that can reach all other alternatives via a domination path of length at most two;\footnote{This is known in graph theory as the set of \emph{kings} (cf. Footnote~\ref{footnote:kings}). An alternative definition, which is also the origin of the name ``uncovered set'', is based on the \emph{covering relation}. An alternative $a_i$ is said to \emph{cover} another alternative $a_j$ if (i) $a_i$ dominates $a_j$, and (ii) any alternative that dominates $a_i$ also dominates $a_j$. The uncovered set corresponds
to the set of alternatives that are not covered by any other alternative.}
\item The \emph{iterated uncovered set} ($\uc^\infty$) is the result of iteratively computing the uncovered set until there is no further reduction.
\end{itemize}
The inclusions $\uc^\infty(T)\subseteq \uc(T)\subseteq \tc(T)\subseteq \cnl(T)$ and $\tc(T)\subseteq \cond(T)$ hold for any tournament $T$.

Next, we describe the random models for generating tournaments that we consider in this paper. We will work with the first model in Sections~\ref{sec:TC} and \ref{sec:UC} and the second model in Section~\ref{sec:extensions}.
\begin{itemize}
\item Model 1: For each pair of distinct alternatives $a_i,a_j$, there is an edge from $a_i$ to $a_j$ with probability $p_{i,j}$ and an edge from $a_j$ to $a_i$ with probability $p_{j,i}=1-p_{i,j}$, independently of other pairs of alternatives.
\item Model 2: There is a constant number $k$ of voters, where $k$ is odd. For each voter $v$ and each pair of distinct alternatives $a_i,a_j$, the voter prefers $a_i$ to $a_j$ with probability $q_{v,i,j}$ and prefers $a_j$ to $a_i$ with probability $q_{v,j,i}=1-q_{v,i,j}$, independently of other voters and other pairs of alternatives.\footnote{One way to interpret the possible intransitivity of the preferences is as a result of noise in the voters' true preferences. \cite{Lasl10a} introduced the term \emph{Rousseauist cultures} for this kind of models.} The majority relation, in which alternative $a_i$ dominates another alternative $a_j$ if and only if more than half of the voters prefer $a_i$ to $a_j$, forms a tournament with $A$ as its set of alternatives.
\end{itemize}
Several models for generating random tournaments considered in previous work are special cases of our models. For example, the \emph{uniform random model} \citep{Fey08a,ScFe11a} corresponds to taking $p_{i,j}=1/2$ for all $i,j$ in Model 1 or taking $q_{v,i,j}=1/2$ for all $v,i,j$ in Model 2 with any $k$. The \emph{Condorcet random model} \citep{Fran68a,LRG96a,Will10a,KSV17a} corresponds to taking $p_{i,j}=p$ for all $i<j$ in Model 1, for some fixed value of $p$. The \emph{Condorcet random model for voters} \citep{BrSe15a} corresponds to taking $q_{v,i,j}=p$ for all $v$ and all $i<j$ in Model 2, for some fixed value of $p$. Following standard terminology, we say that an event occurs ``with high probability'' or ``almost surely'' if the probability that the event occurs converges to 1 as $n$, the number of alternatives, goes to infinity.

We end this section by listing some standard tools for deriving probabilistic bounds. Our first lemma is the Chernoff bound, which gives us an upper bound on the probability that a sum of independent random variables is far away from its expected value.
\begin{lemma}[Chernoff bound] \label{lem:chernoff}
Let $X_1,X_2, \dots, X_r$ be independent random variables that take on values in the interval $[0, 1]$, and let $S=X_1 + X_2 + \cdots + X_r$. For every $\delta\geq 0$, we have
\[\Pr[S \geq (1 + \delta)\E[S]] \leq \exp{\left(\frac{-\delta^2\E[S]}{3}\right)}.\]
\end{lemma}
The next two lemmas allow us to estimate the expression $1+x$ from above and below.
\begin{lemma}[Bernoulli's inequality] \label{lem:bernoulli}
For all real numbers $r\geq 1$ and $x\geq -1$, we have
\[(1+x)^r \geq 1+rx.\]
\end{lemma}
\begin{lemma}
\label{lem:oneplusx}
For all real numbers $x$, we have $1+x\leq e^x$.
\end{lemma}

\section{Top Cycle}
\label{sec:TC}

In this section, we consider the top cycle. We show that when each probability $p_{i,j}$ is between $f(n)$ and $1-f(n)$ for some function $f(n)\in\omega\left(1/n\right)$, $\tc$ chooses all alternatives with high probability (Theorem~\ref{thm:TCall}). By using the inclusion relationships between $\tc$, \cond, and \cnl, we obtain analogous statements for \cond and \cnl. We also show that our results are asymptotically tight---for all three tournament solutions, the statement ceases to hold if $f(n)\in O\left(1/n\right)$ (Theorem~\ref{thm:TCnotall}). 

We begin with our main result of the section.

\begin{theorem}
\label{thm:TCall}
Let $f:\mathbb{Z}^+\rightarrow\mathbb{R}_{\geq 0}$ be a function such that $f(n)\leq 1/2$ for all $n$ and $f(n)\in\omega\left(1/n\right)$. Assume that a tournament $T$ is generated according to Model 1, and that
\[p_{i,j}\in\left[f(n),1-f(n)\right]\] 
for all $i\neq j$. Then with high probability, $\tc(T)=A$.
\end{theorem}

Theorem~\ref{thm:TCall} generalizes a result by \cite{LRG96a} that establishes the claim for the case where $p_{i,j}=f(n)$ for all $i<j$ (or, by symmetry, the case where $p_{i,j}=1-f(n)$ for all $i<j$). We remark that their proof relies crucially on the assumption that there is a linear order of alternatives and all edges are more likely to be oriented in one direction than in the other direction according to the order. Indeed, this assumption allows the authors to show that with high probability, any alternative can be reached by the strongest alternative and can reach the weakest alternative via a domination path of length at most two each. Moreover, with the assumption $f(n)\in\omega\left(1/n\right)$ one can show that the weakest alternative can almost surely reach the strongest alternative via a domination path of length four, thus establishing the strong connectivity of the tournament. In contrast, we do not assume that the edges in the tournament are likely to be oriented in one direction or the other. As such, we will need a completely different approach for our proof.

Before we go into the proof of Theorem~\ref{thm:TCall}, we first give a high-level overview. We observe that $\tc(T)\neq A$ exactly when there exists a proper, nontrivial subset of alternatives $B$ that dominates the complement set of alternatives $A\backslash B$. Using the union bound, we then upper bound the probability that $\tc(T)\neq A$ by the sum over all sets $B$ of the probabilities that $B$ dominates $A\backslash B$. This sum can be written entirely in terms of the variables $p_{i,j}$ for $i<j$ and is moreover linear in all of these variables, implying that its maximum is attained when all variables take on a value at one of the two boundaries of their domain. Using a number of helper lemmas (Lemmas~\ref{lem:karamata}, \ref{lem:majorization-sum}, and \ref{lem:tournament-majorize}), we show that the sum is in fact maximized when all variables take on a value at the same boundary. This allows us to bound the sum directly by plugging in the value at a boundary and complete the proof.

In what follows, we assume that $\textbf{x}=(x_1,x_2,\dots,x_n)$ and $\textbf{y}=(y_1,y_2,\dots,y_n)$ are vectors of nonnegative integers with $n$ components. We start by defining majorization, a preorder on vectors that we will use frequently in our proof.

\begin{definition}[Majorization]
\label{def:majorization}
For a vector $\textbf{x}$, let $\textbf{x}^\downarrow=(x_1^\downarrow,x_2^\downarrow,\dots,x_n^\downarrow)$ be the vector with the same components, but sorted in descending order. Given two vectors $\textbf{x},\textbf{y}$, we say that $\textbf{x}$ \emph{majorizes} $\textbf{y}$, and write $\textbf{x}\succ\textbf{y}$, if the following two conditions are satisfied:
\begin{enumerate}[label=(\roman*)]
\item $\sum_{i=1}^j x_i^\downarrow\geq \sum_{i=1}^j y_i^\downarrow$ for $j=1,2,\dots,n-1$;
\item $\sum_{i=1}^n x_i=\sum_{i=1}^n y_i$.
\end{enumerate}
\end{definition}

When one vector majorizes another vector, Karamata's inequality allows us to compare the sum of an arbitrary convex function at the components of one vector to the corresponding sum of the other vector.

\begin{lemma}[Karamata's inequality]
\label{lem:karamata}
Let $f:\mathbb{Z}_{\geq 0}\rightarrow\mathbb{R}$ be a convex function, and let $\textbf{x},\textbf{y}$ be vectors with $n$ components such that $\textbf{x}\succ\textbf{y}$. Then 
\[\sum_{i=1}^n f(x_i) \geq \sum_{i=1}^n f(y_i).\]
\end{lemma}

We next show that if one vector majorizes another vector, then an analogous statement holds for the two vectors that arise from taking the sum of all subsets with any fixed number of components of the original vectors.

\begin{definition}
Let $n$ be a positive integer and $k\in\{1,2,\dots,n\}$. For a vector $\textbf{x}$ with $n$ components, define $\textbf{x}^{(k)}$ to be the vector with $\binom{n}{k}$ components consisting of all sums of $k$ distinct components of $\textbf{x}$ in nonincreasing order.
\end{definition}

For example, if $n=4$ and $\textbf{x}=(2,4,5,7)$, then $\textbf{x}^{(2)}=(12,11,9,9,7,6)$.

\begin{lemma}
\label{lem:majorization-sum}
If two vectors $\textbf{x},\textbf{y}$ with $n$ components are such that $\textbf{x}\succ\textbf{y}$, then we also have $\textbf{x}^{(k)}\succ \textbf{y}^{(k)}$ for all $k=1,2,\dots,n$.
\end{lemma}

For the sake of continuity, we leave the proof of Lemma~\ref{lem:majorization-sum} along with that of the next lemma to the appendix.

Our final lemma shows that the outdegree vector of a transitive tournament majorizes the corresponding vector of any tournament. Given a tournament $T$ and alternative $a$ in the tournament, denote by $\deg_T(a)$ the outdegree of $a$ in $T$.

\begin{lemma}
\label{lem:tournament-majorize}
Let $W$ be a transitive tournament on alternatives $d_1,d_2,\dots,d_n$, where $d_i$ dominates $d_j$ for all $i<j$. For any tournament $U$ with alternatives $b_1,b_2,\dots,b_n$, we have
\[(\deg_W(d_1),\deg_W(d_2),\dots,\deg_W(d_n))\succ (\deg_U(b_1),\deg_U(b_2),\dots,\deg_U(b_n)).\]
\end{lemma}

With Lemmas~\ref{lem:karamata}, \ref{lem:majorization-sum}, and \ref{lem:tournament-majorize} in hand, we are now ready to prove Theorem~\ref{thm:TCall}.

\begin{proof}[Proof of Theorem~\ref{thm:TCall}]
Let $c\geq 10$ be a constant. Since $f(n)=\omega\left(1/n\right)$, there exists $N'$ such that $f(n)\geq c/n$ for all $n\geq N'$. Let $N=\max(N',4c)$. We will show that for $n\geq N$, the probability that \tc does not choose the whole set of alternatives is at most $16ce^{-\frac{c}{2}}$. Since the expression converges to 0 as $c$ approaches infinity, this will establish the desired result.

Assume that $n\geq N$. Observe that $\tc(T)\neq A$ exactly when there is a proper, nontrivial set of alternatives that dominate the complement set of alternatives. Hence
\begin{align*}
\Pr[\tc(T)\neq A]
&= \Pr[B\text{ dominates } A\backslash B \text{ for some } \emptyset\neq B\subset A] \\
&\leq \sum_{k=1}^{n-1}\sum_{\substack{B\subset A \\ |B|=k}} \Pr[B\text{ dominates } A\backslash B] \\
&= \sum_{k=1}^{n-1}\sum_{\substack{B\subset A \\ |B|=k}}\prod_{\substack{1\leq i\neq j\leq n \\ a_i\in B \\ a_j\in A\backslash B}} p_{i,j} \\
&= \sum_{k=1}^{n-1}\sum_{\substack{B\subset A \\ |B|=k}}\left(\prod_{\substack{1\leq i<j\leq n \\ a_i\in B \\ a_j\in A\backslash B}} p_{i,j}\prod_{\substack{1\leq i<j\leq n \\ a_i\in A\backslash B \\ a_j\in B}} \left(1-p_{i,j}\right)\right), \addtocounter{equation}{1}\tag{\theequation} \label{eqn1}
\end{align*}
where we use the union bound for the inequality. 
We will derive an upper bound for expression (\ref{eqn1}). Note that if we view the terms $p_{i,j}$ with $i<j$ as variables, then the expression in linear in each variable. This implies that the maximum of expression (\ref{eqn1}) over the range $p_{i,j}\in\left[c/n,1-c/n\right]$ is attained when each $p_{i,j}$ is either $c/n$ or $1-c/n$ (but not necessarily when all $p_{i,j}$ are identical). We henceforth assume that for each $i<j$, either $p_{i,j}=c/n$ or $p_{i,j}=1-c/n$. We will show that expression (\ref{eqn1}) is maximized when $p_{i,j}=1-c/n$ for all $i<j$ (or alternatively, when $p_{i,j}=c/n$ for all $i<j$). In fact, we will show the stronger statement that for each particular value of $k$ in the outermost summation, the expression inside the outermost summation is also maximized when $p_{i,j}=1-c/n$ for all $i<j$.

Fix $k\in\{1,2,\dots,n-1\}$. Define a tournament $U$ on $n$ alternatives $b_1,b_2,\dots,b_n$ as follows: For any $i,j\in\{1,2,\dots,n\}$ with $i<j$, there is an edge from $b_i$ to $b_j$ if $p_{i,j}=1-c/n$ and an edge from $b_j$ to $b_i$ if $p_{i,j}=c/n$. We have
\begin{align*}
&\sum_{\substack{B\subset A \\ |B|=k}}\left(\prod_{\substack{1\leq i<j\leq n \\ a_i\in B \\ a_j\in A\backslash B}} p_{i,j}\prod_{\substack{1\leq i<j\leq n \\ a_i\in A\backslash B \\ a_j\in B}} \left(1-p_{i,j}\right)\right) \\
&= \sum_{\substack{B\subset A \\ |B|=k}}\prod_{\substack{1\leq i\neq j\leq n \\ a_i\in B \\ a_j\in A\backslash B}} p_{i,j} \\
&= \left(1-\frac{c}{n}\right)^{k(n-k)} \sum_{\substack{B\subset A \\ |B|=k}} \left(\frac{c}{n-c}\right)^{|\{(a_i,a_j)\in B\times A\backslash B \text{ such that } (b_i,b_j)\text{ is not an edge in } U\}|} \\
&= \left(1-\frac{c}{n}\right)^{k(n-k)} \sum_{\substack{B\subset A \\ |B|=k}} \left(\frac{c}{n-c}\right)^{(n-1)k-\binom{k}{2}-\sum_{i:a_i\in B}\deg_U(b_i)} \\
&= \left(1-\frac{c}{n}\right)^{k(n-k)}\left(\frac{c}{n-c}\right)^{(n-1)k-\binom{k}{2}}\sum_{\substack{B\subset A \\ |B|=k}} \left(\frac{n-c}{c}\right)^{\sum_{i:a_i\in B}\deg_U(b_i)}. \addtocounter{equation}{1}\tag{\theequation} \label{eqn2}
\end{align*}

Let $W$ be a transitive tournament on $n$ alternatives $d_1,d_2,\dots,d_n$, where $d_i$ dominates $d_j$ for all $i<j$. In particular, $\deg_W(d_i)=n-i$ for all $i=1,2,\dots,n$. To show that expression (\ref{eqn1}) is maximized when $p_{i,j}=1-c/n$ for all $i<j$, it suffices to show that expression (\ref{eqn2}) is maximized when $U=W$. The terms outside the summation do not depend on the tournament $U$ that we choose, so for the purpose of maximizing expression (\ref{eqn2}) we may ignore them. 

From Lemma~\ref{lem:tournament-majorize}, we know that
\[(\deg_W(d_1),\deg_W(d_2),\dots,\deg_W(d_n))\succ(\deg_U(b_1),\deg_U(b_2),\dots,\deg_U(b_n)).\]
Lemma~\ref{lem:majorization-sum} then implies that
\[(\deg_W(d_1),\deg_W(d_2),\dots,\deg_W(d_n))^k\succ(\deg_U(b_1),\deg_U(b_2),\dots,\deg_U(b_n))^k.\]
Using Lemma~\ref{lem:karamata} with the convex function $f(x)=\left(\frac{n-c}{c}\right)^x$, we find that 
\[\sum_{\substack{B\subset A \\ |B|=k}} \left(\frac{n-c}{c}\right)^{\sum_{i:a_i\in B}\deg_U(b_i)} \leq \sum_{\substack{B\subset A \\ |B|=k}} \left(\frac{n-c}{c}\right)^{\sum_{i:a_i\in B}\deg_W(d_i)}.\]
It follows that expression (\ref{eqn2}) is maximized when $U=W$, as claimed.

We return to expression (\ref{eqn1}), which we now know is maximized when $p_{i,j}=1-c/n$ for all $i<j$. Substituting $p_{i,j}=1-c/n$ for all $i<j$, expression (\ref{eqn1}) becomes
\begin{align*}
&\sum_{k=1}^{n-1}\left(\left(1-\frac{c}{n}\right)^{k(n-k)} \sum_{\substack{B\subset A \\ |B|=k}} \left(\frac{c}{n-c}\right)^{(n-1)k-\binom{k}{2}-\sum_{i:a_i\in B}\deg_W(d_i)}\right) \\
&\leq 2 \sum_{k=1}^{\lfloor n/2\rfloor} \left(e^{-\frac{ck(n-k)}{n}} \sum_{\substack{B\subset A \\ |B|=k}} \left(\frac{2c}{n}\right)^{(n-1)k-\binom{k}{2}-\sum_{i:a_i\in B}(n-i)}\right) \\
&\leq 2 \sum_{k=1}^{\lfloor n/2\rfloor} \left(e^{-\frac{ck}{2}} \sum_{\substack{B\subset A \\ |B|=k}} \left(\frac{2c}{n}\right)^{\sum_{i:a_i\in B}i-\binom{k+1}{2}}\right), 
\end{align*}
where we use Lemma~\ref{lem:oneplusx}, the assumption $n\geq 4c$, and the symmetry between the terms with $k=i$ and $k=n-i$ for the first inequality. Observe that $\sum_{i:a_i\in B}i-\binom{k+1}{2}$ is always nonnegative, and is zero exactly when $B=\{1,2,\dots,k\}$. Moreover, for any $j=\{1,2,\dots,k\}$, the number of subsets $B\subset A$ with $|B|=k$ such that $\sum_{i:a_i\in B}i-\binom{k+1}{2}\leq j$ is at most $n^j$. Indeed, if a subset $B$ satisfies this inequality, the $n-j$ smallest elements of $B$ must be $1,2,\dots,n-j$, which leaves at most $n^j$ choices for the remaining elements. Note also that $|\{B\subset A\mid |B|=k\}|=\binom{n}{k}\leq n^k$. We have
\begin{align*}
& 2 \sum_{k=1}^{\lfloor n/2\rfloor} \left(e^{-\frac{ck}{2}} \sum_{\substack{B\subset A \\ |B|=k}} \left(\frac{2c}{n}\right)^{\sum_{i:a_i\in B}i-\binom{k+1}{2}}\right) \\
&\leq 2 \sum_{k=1}^{\lfloor n/2\rfloor} \left(e^{-\frac{ck}{2}} \sum_{\substack{B\subset A \\ |B|=k}} \left(\frac{2c}{n}\right)^{\min\left(k,\sum_{i:a_i\in B}i-\binom{k+1}{2}\right)}\right) \\
&\leq 2 \sum_{k=1}^{\lfloor n/2\rfloor} \left(e^{-\frac{ck}{2}} \sum_{j=0}^k \left( n^j\cdot\left(\frac{2c}{n}\right)^j \right) \right) \\
&= 2 \sum_{k=1}^{\lfloor n/2\rfloor} \left(e^{-\frac{ck}{2}} \sum_{j=0}^k (2c)^j\right) \\
&\leq 2 \sum_{k=1}^{\lfloor n/2\rfloor} e^{-\frac{ck}{2}} (4c)^k \\
&= 2 \sum_{k=1}^{\lfloor n/2\rfloor} \left(4ce^{-\frac{c}{2}}\right)^k \\
&\leq 2 \sum_{k=1}^{\infty} \left(4ce^{-\frac{c}{2}}\right)^k \\
&= \frac{8ce^{-\frac{c}{2}}}{1-4ce^{-\frac{c}{2}}} \\
&\leq 16ce^{-\frac{c}{2}},
\end{align*}
where we use the assumption $c\geq 10$ for the last inequality.

In conclusion, when $n\geq N$, the probability that $\tc(T)\neq A$ is at most $16ce^{-\frac{c}{2}}$, completing our proof.
\end{proof}

Since $\tc(T)\subseteq \cond(T)$ and $\tc(T)\subseteq \cnl(T)$, we immediately obtain the following corollary.

\begin{corollary}
\label{cor:COND-CNLall}
Let $f:\mathbb{Z}^+\rightarrow\mathbb{R}_{\geq 0}$ be a function such that $f(n)\leq 1/2$ for all $n$ and $f(n)\in\omega\left(1/n\right)$. Assume that a tournament $T$ is generated according to Model 1, and that
\[p_{i,j}\in\left[f(n),1-f(n)\right]\] 
for all $i\neq j$. Then with high probability, $\cond(T)=\cnl(T)=A$.
\end{corollary}

Next, we show that Theorem~\ref{thm:TCall} and Corollary~\ref{cor:COND-CNLall} are tight in the sense that if $f(n)\in O\left(1/n\right)$, the results no longer hold.

\begin{theorem}
\label{thm:TCnotall}
Let $c\geq 0$ be a constant. Assume that a tournament $T$ is generated according to Model 1, and that 
$$p_{i,j}\leq \frac{c}{n}$$
for all $i>j$. Then for large enough $n$, with at least constant probability both $\tc(T)$ and $\cond(T)$ contain a single alternative. Moreover, for large enough $n$, with at least constant probability $\cnl(T)$ does not contain all alternatives.
\end{theorem}

\begin{proof}
The probability that $a_1$ dominates all of the remaining alternatives is at least
$$\left(1-\frac{c}{n}\right)^{n-1}\rightarrow e^{-c}$$
as $n\rightarrow\infty$. When this occurs, both \tc and \cond only choose $a_1$.

An analogous argument shows that $a_n$ is dominated by all of the remaining alternatives with at least constant probability for large enough $n$. When this occurs, \cnl chooses all alternatives except $a_n$.
\end{proof}

Theorems~\ref{thm:TCall} and \ref{thm:TCnotall} and Corollary~\ref{cor:COND-CNLall} allow us to obtain the following corollary on the Condorcet random model.

\begin{corollary}
\label{cor:TC}
Let $f:\mathbb{Z}^+\rightarrow \mathbb{R}_{\geq 0}$ be a function such that $f(n)\leq 1/2$ for all $n$. Assume that a tournament $T$ is generated according to Model 1, and that $p_{i,j}=f(n)$ for all $i>j$. 
\begin{itemize}
\item If $f(n)\in\omega\left(1/n\right)$, then with high probability, $\tc(T)=\cond(T)=\cnl(T)=A$. 
\item If $f(n)\in o\left(1/n\right)$, then with high probability, $\tc(T)$ and $\cond(T)$ contain a single alternative, and $\cnl(T)$ does not contain all alternatives. 
\item If $f(n)\leq c/n$ for some constant $c\geq 0$, then for large enough $n$, with at least constant probability $\tc(T)$ and $\cond(T)$ contain a single alternative. Moreover, for large enough $n$, with at least constant probability $\cnl(T)$ does not contain all alternatives.
\end{itemize}
\end{corollary}

\cite{LRG96a} also considered the case where $p_{i,j}=c/n$ for all $i>j$ and showed that the probability that \tc selects all alternatives converges to $(1-e^{-c})^2$ in this special case. Our next theorem establishes an analogous result for \cond and \cnl.

\begin{theorem}
\label{thm:CONDnotall}
Let $c\geq 0$ be a constant. Assume that a tournament $T$ is generated according to Model 1, and that 
$$p_{i,j} = \frac{c}{n}$$
for all $i>j$. Then the probability that $\cond(T)=A$ converges to $1-e^{-c}$ as $n\rightarrow\infty$. The same statement holds for \cnl.
\end{theorem}

\begin{proof}
We show the result for \cond; a similar argument holds for \cnl. We have
\begin{align*}
\Pr[\cond(T)\neq A] 
&= \sum_{i=1}^n \Pr[a_i \text{ is a Condorcet winner}] \\
&= \sum_{i=1}^n \left(1-\frac{c}{n}\right)^{n-i}\left(\frac{c}{n}\right)^{i-1} \\
&= \left(1-\frac{c}{n}\right)^{n-1}\cdot\sum_{i=0}^{n-1}\left(\frac{c}{n-c}\right)^i.
\end{align*}
The first term converges to $e^{-c}$ as $n\rightarrow\infty$. For the second term, notice that it is always at least 1. Moreover, when $n\geq (k+1)c$ for some positive $k>1$, the term is at most 
$$1+\frac{1}{k}+\frac{1}{k^2}+\dots=\frac{k}{k-1},$$ which approaches 1 for large $n$. Hence the second term converges to 1, and therefore the probability that $\cond(T)\neq A$ converges to $e^{-c}$, yielding the desired result.
\end{proof}

\section{Uncovered Set}
\label{sec:UC}

In this section, we turn our focus to the uncovered set. We show that when each probability $p_{i,j}$ is between $f(n)$ and $1-f(n)$ for some function $f(n)\geq c\sqrt{\log n/n}$ with $c>\sqrt{2}$ a constant, \uc chooses all alternatives with high probability (Theorem~\ref{thm:UCall}). As with \tc, we also show that our result is asymptotically tight---if $f(n)\leq 0.6\sqrt{\log n/n}$, the statement no longer holds (Theorem~\ref{thm:UCnotall}). It follows that similar results hold for $\uc^\infty$, implying that $\Theta(\sqrt{\log n/n})$ is the threshold where the two tournament solutions go from almost always choosing all alternatives to excluding at least one alternative with high probability.

Our first result of the section shows that \uc chooses the whole set of alternatives for a wide range of distributions over tournaments.

\begin{theorem}
\label{thm:UCall}
Let $c>\sqrt{2}$ be a constant. Assume that a tournament $T$ is generated according to Model 1, and that
\[p_{i,j}\in\left[c\sqrt{\frac{\log n}{n}},1-c\sqrt{\frac{\log n}{n}}\right]\] 
for all $i\neq j$. Then with high probability, $\uc(T)=A$.
\end{theorem}

\begin{proof}
Choose $N$ such that $\frac{c^2(N-2)}{N}>2$, and let $n\geq N$. Fix a pair of distinct alternatives $a_i,a_j$. We first bound the probability that $a_i$ cannot reach $a_j$ via a domination path of length at most two. For each $l\not\in\{i,j\}$, the probability that there is an edge from $a_i$ to $a_l$ and an edge from $a_l$ to $a_j$ is at least $\left(c\sqrt{\log n/n}\right)^2=c^2\log n/n$. The probability that $a_i$ cannot reach $a_j$ via a domination path of length at most two is therefore bounded above by
\begin{align*}
\left(1-\frac{c^2\log n}{n}\right)^{n-2}
&\leq e^{-\frac{c^2(n-2)\log n}{n}}\\ 
&= n^{-\frac{c^2(n-2)}{n}},
\end{align*}
where we use Lemma~\ref{lem:oneplusx} for the inequality.

Observe that $\uc(T)=A$ exactly when any alternative can reach any other alternative via a domination path of length at most two. Using the union bound over all (ordered) pairs of distinct alternatives $i,j$, we find that the probability that some alternative cannot reach some other alternative via a domination path of length at most two is no more than
\[n(n-1)n^{-\frac{c^2(n-2)}{n}}\leq n^{2-\frac{c^2(n-2)}{n}},\]
which vanishes for large $n$.
\end{proof}

Since the uncovered set is the finest tournament solution satisfying the axioms of Condorcet consistency, neutrality, and expansion \citep{Moul86a}, Theorem~\ref{thm:UCall} implies that any tournament solution that satisfies these three axioms also selects all alternatives with high probability when the tournament is generated according to the assumptions of the theorem.

Next, we show that the statement of Theorem~\ref{thm:UCall} breaks down if $f(n)\leq 0.6\sqrt{\log n/n}$, thus confirming that the assumption of the theorem cannot be relaxed asymptotically.

\begin{theorem}
\label{thm:UCnotall}
Assume that a tournament $T$ is generated according to Model 1, and that $$p_{i,j}\leq 0.6\sqrt{\frac{\log n}{n}}$$ for all $i>j$. Then with high probability, $\uc(T)\neq A$.
\end{theorem}

\begin{proof}
Let $A_1=\{a_1,a_2,\dots,a_{\left\lfloor n^{0.49}\right\rfloor}\}$, and let $A_2$ be the set of alternatives that $a_n$ dominates. We first prove the following claim.

\emph{Claim}: With high probability, the following two events occur simultaneously: (i) $a_n$ does not dominate any of the alternatives in $A_1$, and (ii) $|A_2|\leq 0.61\sqrt{n\log n}$. 

\emph{Proof of Claim}: First, using Lemma~\ref{lem:bernoulli}, the probability that $a_n$ dominates at least one of the alternatives in $A_1$ is at most
\begin{align*}
 1-\left(1-0.6\sqrt{\frac{\log n}{n}}\right)^{\left\lfloor n^{0.49}\right\rfloor} 
&\leq 1-\left(1-0.6\left\lfloor n^{0.49}\right\rfloor\sqrt{\frac{\log n}{n}}\right)\\ 
&\leq 0.6\cdot\frac{\sqrt{\log n}}{n^{0.01}},
\end{align*}
which converges to 0 as $n\rightarrow\infty$. 

Next, for each $a_i\in A$ with $i=1,2,\dots,n-1$, let $X_i$ be an indicator random variable that indicates whether $a_n$ dominates $a_i$ or not: $X_i$ takes on the value 1 if $a_n$ dominates $a_i$ and 0 otherwise. We have 
$$\E[X_i]\leq 0.6\sqrt{\frac{\log n}{n}}\leq \frac{0.6\sqrt{n\log n}}{n-1}.$$
Define $X'_i=X_i+\frac{0.6\sqrt{n\log n}}{n-1}-\E[X_i]$, and $X'=\sum_{i=1}^{n-1} X'_i$. We now have 
$$\E[X'_i]=\frac{0.6\sqrt{n\log n}}{n-1} \text{ and } \E[X']=0.6\sqrt{n\log n}.$$
Moreover, observe that $|A_2|=\sum_{i=1}^{n-1} X_i$. By Lemma~\ref{lem:chernoff}, it follows that
\begin{align*}
\Pr\left[|A_2| > 0.61\sqrt{n\log n}\right] 
&\leq \Pr\left[X' > 0.61\sqrt{n\log n}\right] \\
&= \Pr\left[X' > \frac{0.61}{0.6}\cdot\E[X']\right] \\
&\leq \exp\left(-\left(\frac{0.61}{0.6}-1\right)^2\cdot\frac{\E[X']}{3}\right) \\
&= \exp\left(-\frac{1}{3600}\cdot\frac{0.6\sqrt{n\log n}}{3}\right),
\end{align*}
which again vanishes for large $n$. 

Using the union bound over the two events, we have our claim. \hfill $\square$

From now on, we assume that $a_n$ does not dominate any of the alternatives in $A_1$ and that $|A_2|\leq 0.61\sqrt{n\log n}$.  Under this assumption, $a_n$ can reach all of the alternatives in $A_1$ via a domination path of length at most two if and only if each alternative in $A_1$ is dominated by some alternative in $A_2$. Note that the event that this holds for a particular alternative in $A_1$ is independent of the corresponding events for other alternatives in $A_1$. It follows that
\begin{align*}
&\Pr\left[a_n \text{ can reach all } a_i\in A_1 \text{ via a domination path of length at most two}\right] \\
&= \Pr\left[a_n \text{ can reach a fixed } a_i\in A_1 \text{ via a domination path of length two}\right]^{\left\lfloor n^{0.49}\right\rfloor} \\
&= (1-\Pr\left[\text{a fixed } a_i\in A_1 \text{ dominates } a_j \text{ for all } a_j\in A_2\right])^{\left\lfloor n^{0.49}\right\rfloor} \\
&\leq \left(1-\left(1-0.6\sqrt{\frac{\log n}{n}}\right)^{0.61\sqrt{n\log n}}\right)^{\left\lfloor n^{0.49}\right\rfloor} \\
&= \left(1-\left(\left(1-0.6\sqrt{\frac{\log n}{n}}\right)^{0.61\sqrt{\frac{n}{\log n}}}\right)^{\log n}\right)^{\left\lfloor n^{0.49}\right\rfloor} \\
&\leq \left(1-\left(1-0.6\sqrt{\frac{\log n}{n}}\cdot 0.61\sqrt{\frac{n}{\log n}}\right)^{\log n}\right)^{\left\lfloor n^{0.49}\right\rfloor} \\
&= \left(1-0.634^{\log n}\right)^{\left\lfloor n^{0.49}\right\rfloor} \\
&\leq \left(1-n^{-0.46}\right)^{\left\lfloor n^{0.49}\right\rfloor} \\
&\leq e^{-n^{-0.46}\left\lfloor n^{0.49}\right\rfloor},
\end{align*}
where we use Lemma~\ref{lem:bernoulli}, the estimate $0.634>e^{-0.46}$, and Lemma~\ref{lem:oneplusx} for the second, third, and fourth inequalities, respectively. 

Finally, since 
\[\lim_{n\rightarrow\infty}e^{-n^{-0.46}\lfloor n^{0.49}\rfloor}=0,\]
the probability that $a_n\not\in \uc(T)$ converges to 1 as $n$ goes to infinity. This implies that with high probability, $\uc(T)$ is not the whole set of alternatives, as desired.
\end{proof}

Since $\uc(T)=A$ exactly when $\uc^\infty(T)=A$, we immediately have the following corollary.

\begin{corollary}
\label{cor:UCinf}
Assume that a tournament $T$ is generated according to Model 1. 
\begin{itemize}
\item Let $c>\sqrt{2}$ be a constant. If
$p_{i,j}\in\left[c\sqrt{\frac{\log n}{n}},1-c\sqrt{\frac{\log n}{n}}\right]$ 
for all $i\neq j$, then with high probability, $\uc^\infty(T)=A$.
\item If $p_{i,j}\leq 0.6\sqrt{\frac{\log n}{n}}$ for all $i>j$, then with high probability, $\uc^\infty(T)\neq A$.
\end{itemize}
\end{corollary}

Theorems~\ref{thm:UCall} and \ref{thm:UCnotall} and Corollary~\ref{cor:UCinf} allow us to obtain the following corollary on the Condorcet random model.

\begin{corollary}
\label{cor:UC}
Let $f:\mathbb{Z}^+\rightarrow \mathbb{R}_{\geq 0}$ be a function such that $f(n)\leq 1/2$ for all $n$. Assume that a tournament $T$ is generated according to Model 1, and that $p_{i,j}=f(n)$ for all $i>j$. 
\begin{itemize}
\item If $f(n)\in\omega\left(\sqrt{\log n/n}\right)$ or $f(n)\geq c\sqrt{\log n/n}$ for some constant $c>\sqrt{2}$, then with high probability, $\uc(T)=\uc^{\infty}(T)=A$. 
\item If $f(n)\in o\left(\sqrt{\log n/n}\right)$ or $f(n)\leq 0.6\sqrt{\log n/n}$, then with high probability, $\uc(T)\neq A$ and $\uc^{\infty}(T)\neq A$. 
\end{itemize}
\end{corollary}

\section{Majority Tournaments}
\label{sec:extensions}

Thus far, we have established probabilistic results for a general model in which the distribution over tournaments is defined by the probabilities that an alternative dominates another alternative in the tournament (Model 1). As we mentioned in Section~\ref{sec:prelims}, a common interpretation of tournaments is as the majority relation of an odd number of voters who are endowed with linear preferences over a set of alternatives. In this section, we investigate a more specific model in which the distribution over tournaments is determined by the probability that a voter prefers an alternative to another alternative (Model 2). It turns out that the generality of our results for Model 1 will allow us to derive similar results for Model 2 as consequences. 

We first consider the coarser tournament solutions \tc, \cond, and \cnl.

\begin{theorem}
\label{thm:TCallvoting}
Let $f:\mathbb{Z}^+\rightarrow\mathbb{R}_{\geq 0}$ be a function such that $f(n)\leq 1/2$ for all $n$, and $f(n)\in\omega\left(1/n^{2/(k+1)}\right)$. Assume that a tournament $T$ is generated according to Model 2, and that
\[q_{v,i,j} \in [f(n),1-f(n)]\]
for all voters $v$ and all $i\neq j$. Then with high probability, $\tc(T)=\cond(T)=\cnl(T)=A$.
\end{theorem}

\begin{proof}
Since $\tc(T)\subseteq \cond(T)$ and $\tc(T)\subseteq \cnl(T)$, it suffices to prove the statement for \tc. Let $c>0$ be a constant. Since $f(n)\in\omega\left(1/n^{2/(k+1)}\right)$, there exists $N'$ such that $f(n)\geq c/n^{2/(k+1)}$ for all $n\geq N'$. Let $N=\max(N',(2c)^{(k+1)/2})$ and $n\geq N$, and fix a pair of distinct alternatives $a_i,a_j$. Let $p_{i,j}$ denote the probability that $a_i$ dominates $a_j$ in $T$. Observe that $p_{i,j}$ is minimized when $q_{v,i,j}=c/n^{2/(k+1)}$ for all voters $v$. When $q_{v,i,j}$ takes on this value for all $v$, the probability that $a_i$ dominates $a_j$ is at least the probability that exactly $(k+1)/2$ voters prefer $a_i$ to $a_j$. The latter probability is
\begin{align*}
&\binom{k}{\frac{k+1}{2}}\left(\frac{c}{n^{\frac{2}{k+1}}}\right)^{\frac{k+1}{2}}\left(1-\frac{c}{n^{\frac{2}{k+1}}}\right)^{\frac{k-1}{2}} \\
&\geq \binom{k}{\frac{k-1}{2}}\cdot\frac{c^{\frac{k+1}{2}}}{n}\cdot\left(\frac{1}{2}\right)^{\frac{k-1}{2}} \\
&\geq \left(\frac{k}{\frac{k-1}{2}}\right)^{\frac{k-1}{2}}\cdot\frac{c}{n}\cdot\left(\frac{1}{2}\right)^{\frac{k-1}{2}} \\
&\geq 2^{\frac{k-1}{2}}\cdot\frac{c}{n}\cdot\left(\frac{1}{2}\right)^{\frac{k-1}{2}} \\
&\geq \frac{c}{n},
\end{align*}
where we use the assumption $n\geq (2c)^{(k+1)/2}$ for the first inequality and the approximation $\binom{n}{k}\geq\left(n/k\right)^k$ for the second inequality.
Hence $p_{i,j}\geq c/n$ for all $n\geq N$. This implies that there exists a function $f(n)\in\omega\left(1/n\right)$ such that $p_{i,j}\in [f(n),1-f(n)]$. Using Theorem~\ref{thm:TCall}, we have that $\tc(T)=A$ with high probability, as desired. 
\end{proof}

We now consider the finer tournament solutions \uc and $\uc^\infty$.

\begin{theorem}
\label{thm:UCallvoting}
Let $c>\sqrt{2}$ be a constant. Assume that a tournament $T$ is generated according to Model 2, and that
\[q_{v,i,j} \in \left[c\left(\frac{\log n}{n}\right)^{\frac{1}{k+1}},1-c\left(\frac{\log n}{n}\right)^{\frac{1}{k+1}}\right]\]
for all voters $v$ and all $i\neq j$. Then with high probability, $\uc(T)=\uc^\infty(T)=A$.
\end{theorem}

\begin{proof}
Since $\uc(T)=A$ implies that $\uc^\infty(T)=A$, it suffices to prove the statement for \uc. Let $n\geq (2c)^{2k+2}$, and fix a pair of distinct alternatives $a_i,a_j$. Let $p_{i,j}$ denote the probability that $a_i$ dominates $a_j$ in $T$. Observe that $p_{i,j}$ is minimized when $q_{v,i,j}=c\left(\log n/n\right)^{1/(k+1)}$ for all voters $v$. When $q_{v,i,j}$ takes on this value for all $v$, the probability that $a_i$ dominates $a_j$ is at least the probability that exactly $(k+1)/2$ voters prefer $a_i$ to $a_j$. The latter probability is 
\begin{align*}
&\binom{k}{\frac{k+1}{2}}\left(c\left(\frac{\log n}{n}\right)^{\frac{1}{k+1}}\right)^{\frac{k+1}{2}}\left(1-c\left(\frac{\log n}{n}\right)^{\frac{1}{k+1}}\right)^{\frac{k-1}{2}} \\\
&\geq \binom{k}{\frac{k-1}{2}}\left(c\left(\frac{\log n}{n}\right)^{\frac{1}{k+1}}\right)^{\frac{k+1}{2}}\left(\frac{1}{2}\right)^{\frac{k-1}{2}} \\
&\geq \left(\frac{k}{\frac{k-1}{2}}\right)^{\frac{k-1}{2}}\left(c\left(\frac{\log n}{n}\right)^{\frac{1}{k+1}}\right)^{\frac{k+1}{2}}\left(\frac{1}{2}\right)^{\frac{k-1}{2}} \\
&\geq 2^{\frac{k-1}{2}}c^{\frac{k+1}{2}}\sqrt{\frac{\log n}{n}}\left(\frac{1}{2}\right)^{\frac{k-1}{2}} \\
&\geq c\sqrt{\frac{\log n}{n}},
\end{align*}
where we use the assumption $n\geq (2c)^{2k+2}$ for the first inequality and the approximation $\binom{n}{k}\geq\left(n/k\right)^k$ for the second inequality.
Using Theorem~\ref{thm:UCall}, we have that $\uc(T)=A$ with high probability, as desired. 
\end{proof}

\section{Experiments}
\label{sec:experiments}
To complement our theoretical results, in this section we investigate the asymptotic behavior of random tournaments according to the Condorcet random model as well as another more realistic model that we call the \emph{gap model}. 

\subsection{Condorcet Random Model}

Starting from a set of alternatives $\{a_1, a_2, \dots ,a_n \}$,
 we generate random tournaments according to the Condorcet random model by inserting for each pair of alternatives $a_i,a_j$ with $i>j$ an edge from $a_i$ to $a_j$ with probability $p$ and an edge in the reverse direction with probability $1-p$. The tournament solutions that we consider can all be computed efficiently: A simple counting algorithm suffices to compute \cond, a depth-first search algorithm computes \tc in linear time, and the asymptotic running time for computing \uc equals that of matrix multiplication \citep{Hudr09a}.
In our experimental setup, we draw 10000 random tournaments of each size $n \in \{5,10,20,30, \dots,100\}$ for each $p\in\{0.5,0.3,1/n,1/n^2,\sqrt{2\log n/n}, 0.6\sqrt{\log n/n}\}$ and check for each tournament solution $S \in \{\cond,\uc,\tc\}$ whether it selects all alternatives.\footnote{Our setting is slightly different for the last two values of $p$, as we explain later in this section.}$^{,}$\footnote{Since the probability that \cnl selects all alternatives is equal to the corresponding probability for \cond for any fixed $n$ by symmetry, and $\uc^\infty$ selects all alternatives exactly when \uc does, the results for \cnl and $\uc^\infty$ are captured by those for \cond and \uc, respectively.} Out of that, we compute the percentage of tournaments in which all alternatives are selected. The resulting graphs are displayed in Figure~\ref{fig:condorcet}.

% Note that $0.6 \sqrt{\frac{\log 50}{50}} \approx 2*\sqrt{\frac{\log 1000}{1000}} \approx 0.17$,

\begin{figure}
	\centering
	\begin{subfigure}{.5\textwidth} 
		\centering
\begin{tikzpicture}[scale=0.85]
\begin{axis}[
    title={(a) $p=0.5$},
    xlabel={},
    ylabel={},
    xmin=5, xmax=100,
    ymin=0, ymax=100,
    xtick= {5,20,40,60,80,100},
    ytick={20,40,60,80,100},
    legend pos=south east,
    ymajorgrids=true,
		xmajorgrids=true,
    grid style=dashed,
]

\addplot[
    color=black,
	  mark size=3,
    mark=triangle*,
  	fill opacity=0.8,
  	draw opacity=0.8,
    ]
    coordinates {
    (5,69.1)(10,98.4)(20,100)(30,100)(40,100)(50,100)(60,100)(60,100)(70,100)(80,100)(90,100)(100,100)
    };
    \addlegendentry{\cond}

\addplot[
	    color=red,
	    mark=square*,
			fill opacity=0.8,
  	  draw opacity=0.8,
	    ]
	    coordinates {
	    (5,53.5)(10,96.4)(20,100)(30,100)(40,100)(50,100)(60,100)(70,100)(80,100)(90,100)(100,100)
	    };
	    \addlegendentry{\tc}

\addplot[
    color=blue,
    mark=otimes*,
		fill opacity=0.8,
  	draw opacity=0.8,
    ]
    coordinates {
    (5,5.8)(10,7.6)(20,46.4)(30,88.4)(40,98.6)(50,99.9)(60,100)(70,100)(80,100)(90,100)(100,100)
    };
    \addlegendentry{\uc}
\end{axis}
\end{tikzpicture}
\end{subfigure}\begin{subfigure}{.5\textwidth} 
	\centering
\begin{tikzpicture}[scale=0.85]
\begin{axis}[
    title={(b) $p=0.3$},
    xlabel={},
    ylabel={},
    xmin=5, xmax=100,
    ymin=0, ymax=100,
    xtick={5,20,40,60,80,100},
    ytick={20,40,60,80,100},
    legend pos=south east,
    ymajorgrids=true,
		xmajorgrids=true,
    grid style=dashed,
]

\addplot[
    color=black,
    mark size=3,
    mark=triangle*,
  	fill opacity=0.8,
  	draw opacity=0.8,
    ]
    coordinates {
    (5,59.4)(10,92.3)(20,99.8)(30,100)(40,100)(50,100)(60,100)(70,100)(80,100)(90,100)(100,100)
    };
    \addlegendentry{\cond}

\addplot[
	    color=red,
        mark=square*,
        fill opacity=0.8,
        draw opacity=0.8,
	    ]
	    coordinates {
	    (5,42.1)(10,86.3)(20,99.6)(30,100)(40,100)(50,100)(60,100)(70,100)(80,100)(90,100)(100,100)	
    	   };
	    \addlegendentry{\tc}

\addplot[
    color=blue,
    mark=otimes*,
    fill opacity=0.8,
  	draw opacity=0.8,
    ]
    coordinates {
     (5,4.3)(10,2.2)(20,5.9)(30,20.6)(40,44.5)(50,69.6)(60,85.6)(70,93.9)(80,97.4)(90,99.0)(100,99.6)
    };
    \addlegendentry{\uc}
\end{axis}
\end{tikzpicture}
\end{subfigure}

\vspace{2mm}

	\centering
	\begin{subfigure}{.5\textwidth} 
		\centering
\begin{tikzpicture}[scale=0.85]
\begin{axis}[
    title={(c) $p=\frac{1}{n}$},
    xlabel={},
    ylabel={},
    xmin=5, xmax=100,
    ymin=0, ymax=100,
    xtick={5,20,20,40,60,80,100},
    ytick={20,40,60,80,100},
    legend pos=south east,
    ymajorgrids=true,
		xmajorgrids=true,
    grid style=dashed,
]

\addplot[
    color=black,
	  mark size=3,
    mark=triangle*,
  	fill opacity=0.8,
  	draw opacity=0.8,
    ]
    coordinates {
    (5,45.3)(10,56.5)(20,60.7)(30,60.9)(40,61)(50,61.4)(60,61.9)(70,62.1)(80,62.4)(90,62.3)(100,62.5)
    };
    \addlegendentry{\cond}

\addplot[
	    color=red,
        mark=square*,
        fill opacity=0.8,
        draw opacity=0.8,
	    ]
	    coordinates {
    (5,27.8)(10,32.7)(20,37)(30,37)(40,37.9)(50,38.3)(60,38.8)(70,38.9)(80,39)(90,39)(100,39.2)
	    };
	    \addlegendentry{\tc}

\addplot[
    color=blue,
    mark=otimes*,
    fill opacity=0.8,
  	draw opacity=0.8,
    ]
    coordinates {
    (5,0)(10,0)(20,0)(30,0)(40,0)(50,0)(60,0)(70,0)(80,0)(90,0)(100,0)
    };
    \addlegendentry{\uc}
\end{axis}
\end{tikzpicture}
\end{subfigure}\begin{subfigure}{.5\textwidth} 
	\centering
\begin{tikzpicture}[scale=0.85]
\begin{axis}[
    title={(d) $p= \frac{1}{n^2}$},
    xlabel={},
    ylabel={},
    xmin=5, xmax=100,
    ymin=0, ymax=100,
    xtick={5,20,40,60,80,100},
    ytick={20,40,60,80,100},
    legend pos= north east,
    ymajorgrids=true,
	  xmajorgrids=true,
    grid style=dashed,
]

\addplot[
    color=black,
    mark size=3,
    mark=triangle*,
  	fill opacity=0.8,
  	draw opacity=0.8,
    ]
    coordinates {
    (5,11.5)(10,7.1)(20,4.2)(30,3.4)(40,2.6)(50,1.8)(60,1.7)(70,1.3)(80,1.3)(90,1.1)(100,0.8)
    };
    \addlegendentry{\cond}

\addplot[
	    color=red,
        mark=square*,
        fill opacity=0.8,
        draw opacity=0.8,
	    ]
	    coordinates {
		(5,4.8)(10,1.3)(20,0.4)(30,0.1)(40,0.2)(50,0.1)(60,00)(70,0)(80,0)(90,0)(100,0)
		};
	    \addlegendentry{\tc}

\addplot[
    color=blue,
    mark=otimes*,
    fill opacity=0.8,
  	draw opacity=0.8,
    ]
    coordinates {
    (5,0.2)(10,0)(20,0)(30,0)(40,0)(50,0)(60,0)(70,0)(80,0)(90,0)(100,0)
    };
    \addlegendentry{\uc}
\end{axis}
\end{tikzpicture}
\end{subfigure}

\vspace{2mm}

	\begin{subfigure}{.5\textwidth} 
		\centering
\begin{tikzpicture}[scale=0.85]
\begin{axis}[
    title={(e) $p= \sqrt{\frac{2 \log n}{n}}$},
    xlabel={},
    ylabel={},
    xmin=50, xmax=1000,
    ymin=0, ymax=100,
    xtick={50,200,400,600,800,1000},
    ytick={20,40,60,80,100},
    legend pos=south east,
    ymajorgrids=true,
		xmajorgrids=true,
    grid style=dashed,
]

\addplot[
    color=black,
    mark size=3,
    mark=triangle*,
  	fill opacity=0.8,
  	draw opacity=0.8,
    ]
    coordinates {
    (50,100)(100,100)(150,100)(200,100)(300,100)(400,100)(500,100)(600,100)(700,100)(800,100)(900,100)(1000,100)	
    };
    \addlegendentry{\cond}

\addplot[
        color=red,
        mark=square*,
        fill opacity=0.8,
        draw opacity=0.8,
	    ]
	    coordinates {
    (50,100)(100,100)(150,100)(200,100)(300,100)(400,100)(500,100)(600,100)(700,100)(800,100)(900,100)(1000,100)	
		};
	    \addlegendentry{\tc}

\addplot[
    color=blue,
    mark=otimes*,
    fill opacity=0.8,
  	draw opacity=0.8,
    ]
    coordinates {
    (50,100)(100,100)(150,100)(200,100)(300,100)(400,100)(500,100)(600,100)(700,100)(800,100)(900,100)(1000,100)	
    };
    \addlegendentry{\uc}
\end{axis}
\end{tikzpicture}
\end{subfigure}\begin{subfigure}{.5\textwidth} 
	\centering
\begin{tikzpicture}[scale=0.85]
\begin{axis}[
    title={(f) $p= 0.6 \sqrt{\frac{ \log n}{n}}$},
    xlabel={},
    ylabel={},
    xmin=50, xmax=1000,
    ymin=0, ymax=100,
    xtick={50,200,400,600,800,1000},
    ytick={20,40,60,80,100},
    legend pos=south east,
    ymajorgrids=true,
		xmajorgrids=true,
    grid style=dashed,
]

\addplot[
    color=black,
    mark size=3,
    mark=triangle*,
  	fill opacity=0.8,
  	draw opacity=0.8,
    ]
    coordinates {
    (50,100)(100,100)(150,100)(200,100)(300,100)(400,100)(500,100)(600,100)(700,100)(800,100)(900,100)(1000,100)	
    };
    \addlegendentry{\cond}

\addplot[
        color=red,
        mark=square*,
        fill opacity=0.8,
        draw opacity=0.8,
	    ]
	    coordinates {
    (50,100)(100,100)(150,100)(200,100)(300,100)(400,100)(500,100)(600,100)(700,100)(800,100)(900,100)(1000,100)	
		};
	    \addlegendentry{\tc}

\addplot[
    color=blue,
    mark=otimes*,
    fill opacity=0.8,
  	draw opacity=0.8,
    ]
    coordinates {
    (50,0)(100,0)(150,0)(200,0)(300,0)(400,0)(500,0)(600,0)(700,0)(800,0)(900,0)(1000,0)	
    };
    \addlegendentry{\uc}
\end{axis}
\end{tikzpicture}
\end{subfigure}
\caption{Percentage of tournaments for which the tournament solution chooses the entire set of alternatives in the Condorcet random model, for different values of the probability $p$. The horizontal and vertical axes correspond to the number of alternatives in the tournament and the percentage, respectively. Averages are taken over 10000 runs.}
\label{fig:condorcet}
\end{figure}
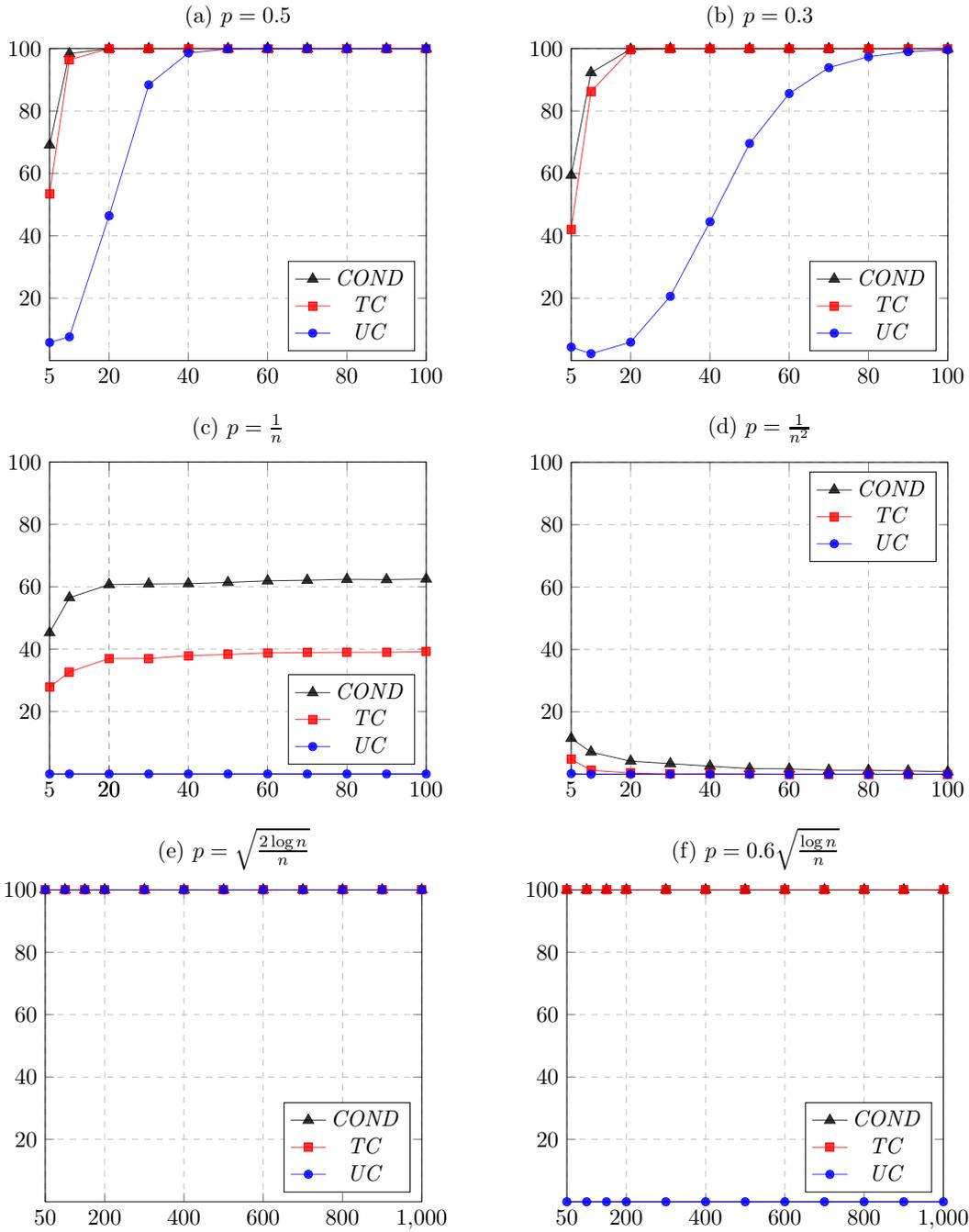

For $p=0.5$, which corresponds to the uniform random model, our experimental results in Figure~\ref{fig:condorcet}(a) coincide with the main theorem of \cite{Fey08a}. The results moreover reveal that \uc chooses all alternatives with high probability in tournaments with at least 50 alternatives while \cond and \tc already do so in much smaller tournaments. As $p$ decreases from 0.5 toward 0, the curves of \cond, \tc, and \uc are shifted to the right; this is to be expected since for smaller $p$ the tournament is more skewed, making it more likely for weaker alternatives to be excluded. Nevertheless, for any fixed $p$ the fraction of tournaments in which all alternatives are chosen approaches 1. In particular, when $p = 0.3$, \uc almost never rules out any alternative in tournaments of size 100 or more (Figure~\ref{fig:condorcet}(b)).

Next, we look at the regimes where the probability $p$ goes to 0 as $n$ approaches infinity. For the case of $p = 1/n$ we find that, in line with Theorem~\ref{thm:CONDnotall}, the probability that \cond selects all alternatives converges to  $1-e^{-1}\approx  0.6321$ (Figure~\ref{fig:condorcet}(c)). Similarly, the probability that \tc selects all alternatives converges to $(1-e^{-1})^2 \approx 0.3996$ for the same value of $p$, confirming a result by \cite{LRG96a}. Letting $p$ approach 0 even faster, we find that for $p = 1/n^2$, both \tc and \cond are discriminative with high probability (Figure~\ref{fig:condorcet}(d)). As $1/n^2 \in o\left(1/n\right)$, this is consistent with Corollary~\ref{cor:TC}. Note that \uc is discriminative for almost all tournaments for both $p=1/n$ and $p=1/n^2$; indeed, this is implied by Corollary~\ref{cor:UC} since already $1/n \in o\left(\sqrt{\log n/n}\right)$.

Finally, we consider the regime $p=\Theta\left(\sqrt{\log n/n}\right)$, which according to Corollary~\ref{cor:UC} is the boundary between \uc almost never ruling out any alternative and almost always ruling out at least one alternative. The experimental setting for $p = c\sqrt{\log n/n}$ with $c \in \{0.6, \sqrt{2} \}$ differs from the previous settings in that we only examined tournaments of size $n \geq 50$, since for small $n$ the expression $\sqrt{2\log n/n}$ is larger than $0.5$, making it unsuitable for our experiments. On the other hand, as $p$ decreases rather slowly, we examined random tournaments up to size 1000 in order to increase the expressive power of our experiments. We find that \cond and \tc select all alternatives with high probability for both values of $c$; this is in line with Corollary~\ref{cor:TC} and the observation that $c\sqrt{\log n/n}\in\omega\left(1/n\right)$.
On the other hand, our experiments indicate that \uc returns all alternatives in almost all tournaments in the case of $p = \sqrt{2\log n/n}$ (Figure \ref{fig:condorcet}(e)) but is discriminative in almost all tournaments when $p = 0.6\sqrt{\log n/n}$ (Figure \ref{fig:condorcet}(f)). These findings coincide with Corollary~\ref{cor:UC} and demonstrate the interesting fact that a small gap in the constant factor constitutes the threshold with regard to the discriminative power of \uc.

\subsection{Gap Model}

As we explained in the introduction, while the Condorcet random model is commonly used in theoretical analyses, it does not properly capture tournaments in the real world since it assigns the same probability to all edges regardless of the difference in strength between the two alternatives adjacent to that edge. We next consider a different model, which we call the \emph{gap model}, that takes this issue into account. 

Like in the Condorcet random model, in the gap model there is a linear order of alternatives from strongest to weakest as well as a parameter $p\leq 0.5$. However, the probability that a stronger alternative dominates a weaker alternative depends linearly on the size of the gap between the two alternatives in the linear order: For two alternatives $a_i,a_j$ with $i<j$, there is an edge from $a_i$ to $a_j$ with probability $0.5 + \frac{(0.5-p)(j-i)}{n-1}$, and an edge in the reverse direction with the remaining probability. In particular, there is an edge from $a_1$ to $a_n$ with probability $1-p$.  We perform experiments on the gap model using the same values of $p$ as we did for the Condorcet random model. The only exception is $p=0.5$ which we replace by $p=0$, the reason being that both models coincide when $p=0.5$. Moreover, we double the sizes of the tournaments considered for the first four values of $p$ in order to better illustrate the convergence behavior. The resulting graphs are displayed in Figure~\ref{fig:gap}.

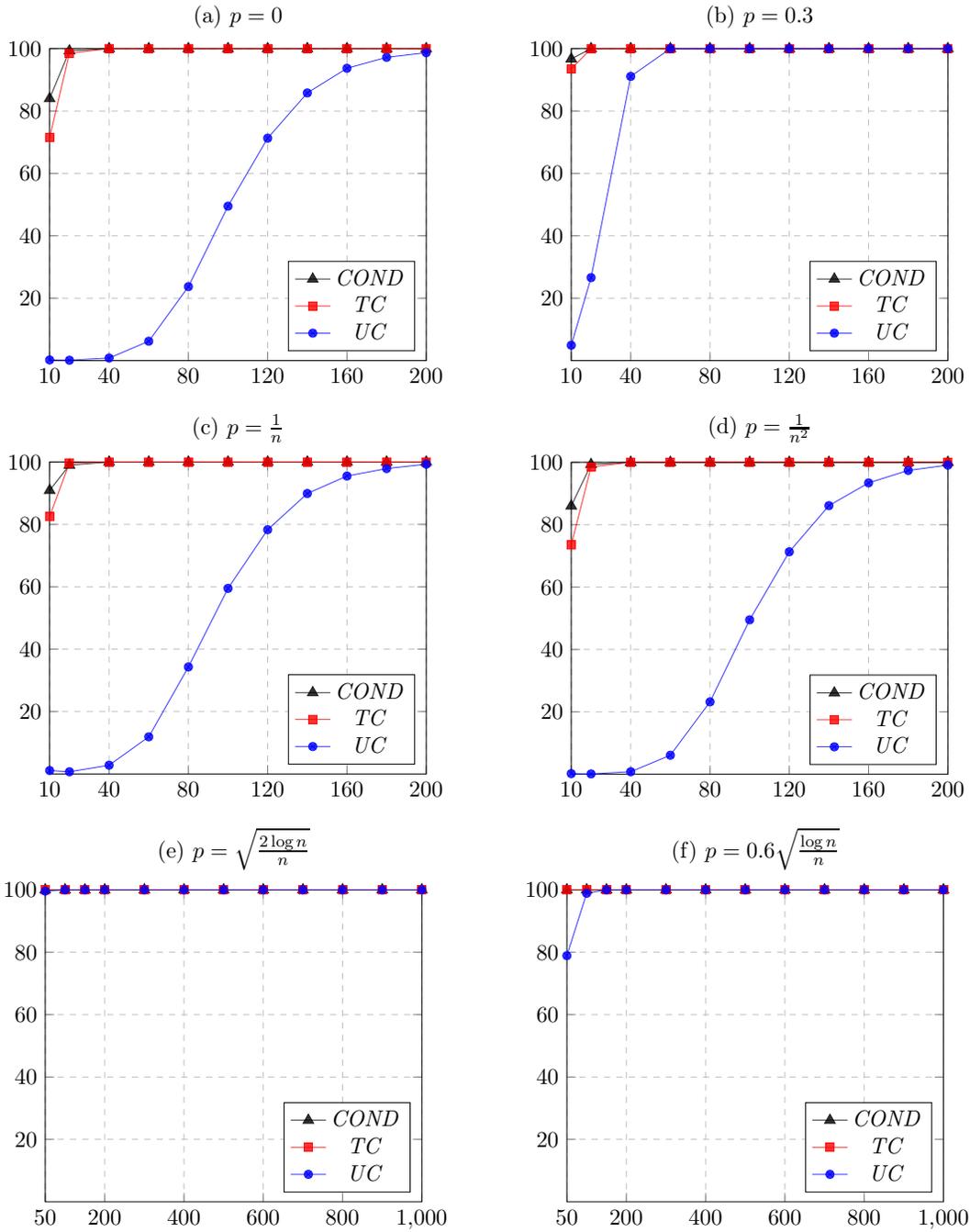
\begin{figure}
	\centering
	\begin{subfigure}{.5\textwidth} 
		\centering
\begin{tikzpicture}[scale=0.85]
\begin{axis}[
    title={(a) $p=0$},
 xlabel={},
 ylabel={},
 xmin=10, xmax=200,
 ymin=0, ymax=100,
 xtick={10,40,80,120,160,200},
 ytick={20,40,60,80,100},
 legend pos=south east,
 ymajorgrids=true,
 xmajorgrids=true,
 grid style=dashed,
]

\addplot[
    color=black,
    mark size=3,
    mark=triangle*,
  	fill opacity=0.8,
  	draw opacity=0.8,
    ]
    coordinates {
    (10,84.0)(20,99.3)(40,100)(60,100)(80,100)(100,100)(120,100)(140,100)(160,100)(180,100)(200,100)
    };
    \addlegendentry{\cond}

\addplot[
        color=red,
        mark=square*,
        fill opacity=0.8,
        draw opacity=0.8,
	    ]
	    coordinates {
		(10,71.5)(20,98.5)(40,100)(60,100)(80,100)(100,100)(120,100)(140,100)(160,100)(180,100)(200,100)
	    };
	    \addlegendentry{\tc}

\addplot[
    color=blue,
    mark=otimes*,
    fill opacity=0.8,
  	draw opacity=0.8,
    ]
    coordinates {
    (10,0.2)(20,0.1)(40,0.8)(60,6.2)(80,23.7)(100,49.5)(120,71.3)(140,85.8)(160,93.7)(180,97.2)(200,98.7)
    };
    \addlegendentry{\uc}
\end{axis}
\end{tikzpicture}
\end{subfigure}\begin{subfigure}{.5\textwidth} 
	\centering
\begin{tikzpicture}[scale=0.85]
\begin{axis}[
    title={(b) $p=0.3$},
    xlabel={},
    ylabel={},
    xmin=10, xmax=200,
    ymin=0, ymax=100,
    xtick={10,40,80,120,160,200},
    ytick={20,40,60,80,100},
    legend pos=south east,
    ymajorgrids=true,
		xmajorgrids=true,
    grid style=dashed,
]

\addplot[
    color=black,
    mark size=3,
    mark=triangle*,
  	fill opacity=0.8,
  	draw opacity=0.8,
    ]
    coordinates {
	(10,96.7)(20,100.0)(40,100)(60,100)(80,100)(100,100)(120,100)(140,100)(160,100)(180,100)(200,100)
    };
    \addlegendentry{\cond}

\addplot[
        color=red,
        mark=square*,
        fill opacity=0.8,
        draw opacity=0.8,
	    ]
	    coordinates {
		(10,93.5)(20,100)(40,100)(60,100)(80,100)(100,100)(120,100)(140,100)(160,100)(180,100)(200,100)
    	   };
	    \addlegendentry{\tc}

\addplot[
    color=blue,
    mark=otimes*,
    fill opacity=0.8,
  	draw opacity=0.8,
    ]
    coordinates {
	(10,4.9)(20,26.6)(40,91.1)(60,100)(80,100)(100,100)(120,100)(140,100)(160,100)(180,100)(200,100)
    };
    \addlegendentry{\uc}
\end{axis}
\end{tikzpicture}
\end{subfigure}

\vspace{2mm}

	\centering
	\begin{subfigure}{.5\textwidth} 
		\centering
\begin{tikzpicture}[scale=0.85]
\begin{axis}[
    title={(c) $p=\frac{1}{n}$},
    xlabel={},
    ylabel={},
    xmin=10, xmax=200,
    ymin=0, ymax=100,
    xtick={10,40,80,120,160,200},
    ytick={20,40,60,80,100},
    legend pos=south east,
    ymajorgrids=true,
		xmajorgrids=true,
    grid style=dashed,
]
\addplot[
    color=black,
    mark size=3,
    mark=triangle*,
  	fill opacity=0.8,
  	draw opacity=0.8,
    ]
    coordinates {
    (10,90.9)(20,99.0)(40,100)(60,100)(80,100)(100,100)(120,100)(140,100)(160,100)(180,100)(200,100)
    };
    \addlegendentry{\cond}

\addplot[
        color=red,
        mark=square*,
        fill opacity=0.8,
        draw opacity=0.8,
	    ]
	    coordinates {
		(10,82.5)(20,99.7)(40,100)(60,100)(80,100)(100,100)(120,100)(140,100)(160,100)(180,100)(200,100)
	    };
	    \addlegendentry{\tc}

\addplot[
    color=blue,
    mark=otimes*,
    fill opacity=0.8,
  	draw opacity=0.8,
    ]
    coordinates {
    (10,1.1)(20,0.7)(40,2.8)(60,11.9)(80,34.3)(100,59.5)(120,78.3)(140,89.9)(160,95.5)(180,97.9)(200,99.3)
    };
    \addlegendentry{\uc}
\end{axis}
\end{tikzpicture}
\end{subfigure}\begin{subfigure}{.5\textwidth} 
	\centering
\begin{tikzpicture}[scale=0.85]
\begin{axis}[
    title={(d) $p= \frac{1}{n^2}$},
 xlabel={},
 ylabel={},
 xmin=10, xmax=200,
 ymin=0, ymax=100,
 xtick={10,40,80,120,160,200},
 ytick={20,40,60,80,100},
 legend pos=south east,
 ymajorgrids=true,
 xmajorgrids=true,
 grid style=dashed,
]

\addplot[
    color=black,
    mark size=3,
    mark=triangle*,
  	fill opacity=0.8,
  	draw opacity=0.8,
    ]
    coordinates {
    (10,86.0)(20,99.3)(40,100)(60,100)(80,100)(100,100)(120,100)(140,100)(160,100)(180,100)(200,100)
    };
    \addlegendentry{\cond}

\addplot[
	    color=red,
        mark=square*,
        fill opacity=0.8,
        draw opacity=0.8,
	    ]
	    coordinates {
		(10,73.5)(20,98.5)(40,100)(60,100)(80,100)(100,100)(120,100)(140,100)(160,100)(180,100)(200,100)
	    };
	    \addlegendentry{\tc}

\addplot[
    color=blue,
    mark=otimes*,
    fill opacity=0.8,
  	draw opacity=0.8,
    ]
    coordinates {
    (10,0.2)(20,0.1)(40,0.8)(60,6.1)(80,23.2)(100,49.5)(120,71.3)(140,86.1)(160,93.4)(180,97.4)(200,99.1)
    };
    \addlegendentry{\uc}
\end{axis}
\end{tikzpicture}
\end{subfigure}

\vspace{2mm}

	\begin{subfigure}{.5\textwidth} 
		\centering
\begin{tikzpicture}[scale=0.85]
\begin{axis}[
    title={(e) $p= \sqrt{\frac{2 \log n}{n}}$},
    xlabel={},
    ylabel={},
    xmin=50, xmax=1000,
    ymin=0, ymax=100,
    xtick={50,200,400,600,800,1000},
    ytick={20,40,60,80,100},
    legend pos=south east,
    ymajorgrids=true,
		xmajorgrids=true,
    grid style=dashed,
]

\addplot[
    color=black,
    mark size=3,
    mark=triangle*,
  	fill opacity=0.8,
  	draw opacity=0.8,
    ]
    coordinates {
    (50,100)(100,100)(150,100)(200,100)(300,100)(400,100)(500,100)(600,100)(700,100)(800,100)(900,100)(1000,100)	
    };
    \addlegendentry{\cond}

\addplot[
        color=red,
        mark=square*,
        fill opacity=0.8,
        draw opacity=0.8,
	    ]
	    coordinates {
    (50,100)(100,100)(150,100)(200,100)(300,100)(400,100)(500,100)(600,100)(700,100)(800,100)(900,100)(1000,100)	
		};
	    \addlegendentry{\tc}

\addplot[
    color=blue,
    mark=otimes*,
    fill opacity=0.8,
  	draw opacity=0.8,
    ]
    coordinates {
    (50,99.6)(100,100)(150,100)(200,100)(300,100)(400,100)(500,100)(600,100)(700,100)(800,100)(900,100)(1000,100)	
    };
    \addlegendentry{\uc}
\end{axis}
\end{tikzpicture}
\end{subfigure}\begin{subfigure}{.5\textwidth} 
	\centering
\begin{tikzpicture}[scale=0.85]
\begin{axis}[
    title={(f) $p= 0.6 \sqrt{\frac{ \log n}{n}}$},
    xlabel={},
    ylabel={},
    xmin=50, xmax=1000,
    ymin=0, ymax=100,
    xtick={50,200,400,600,800,1000},
    ytick={20,40,60,80,100},
    legend pos=south east,
    ymajorgrids=true,
		xmajorgrids=true,
    grid style=dashed,
]

\addplot[
    color=black,
    mark size=3,
    mark=triangle*,
  	fill opacity=0.8,
  	draw opacity=0.8,
    ]
    coordinates {
    (50,100)(100,100)(150,100)(200,100)(300,100)(400,100)(500,100)(600,100)(700,100)(800,100)(900,100)(1000,100)	
    };
    \addlegendentry{\cond}

\addplot[
        color=red,
        mark=square*,
        fill opacity=0.8,
        draw opacity=0.8,
	    ]
	    coordinates {
    (50,100)(100,100)(150,100)(200,100)(300,100)(400,100)(500,100)(600,100)(700,100)(800,100)(900,100)(1000,100)	
		};
	    \addlegendentry{\tc}

\addplot[
    color=blue,
    mark=otimes*,
    fill opacity=0.8,
  	draw opacity=0.8,
    ]
    coordinates {
    (50,78.9)(100,98.9)(150,100)(200,100)(300,100)(400,100)(500,100)(600,100)(700,100)(800,100)(900,100)(1000,100)	
    };
    \addlegendentry{\uc}
\end{axis}
\end{tikzpicture}
\end{subfigure}
\caption{Percentage of tournaments for which the tournament solution chooses the entire set of alternatives in the gap model, for different values of the probability $p$. The horizontal and vertical axes correspond to the number of alternatives in the tournament and the percentage, respectively. Averages are taken over 10000 runs.}
\label{fig:gap}
\end{figure}

We now make some observations. Firstly, in Figure~\ref{fig:gap}(b), we find that for $p=0.3$ all three tournament solutions are unlikely to exclude any alternative in large tournaments; this is consistent with Theorems~\ref{thm:TCall} and \ref{thm:UCall}. The same phenomenon occurs for $p=\sqrt{2\log n/n}$ (Figure~\ref{fig:gap}(e)). We remark that these phenomena cannot be explained by any theoretical result prior to our work. In the remaining figures, we likewise find that in the gap model, all tournament solutions cease to be discriminative as the size of the tournament grows. This is the case even for the extreme value $p=0$ (Figure~\ref{fig:gap}(a)); the Condorcet random model for this value of $p$ clearly always produces a Condorcet winner. Intuitively, the reason that all alternatives are likely to be chosen is that in the gap model, the overall difference in strength between alternatives is significantly less than in the Condorcet random model. Note that the observations for the latter values of $p$ are not captured by our results, since our theorems require all edge probabilities to be in the range $[p,1-p]$ for appropriate values of $p$. Indeed, an intriguing direction for future work would be to generalize our results so that some edge probabilities are allowed to be outside of this range.\footnote{A result of this flavor has been shown by \cite{KSV17a} in the context of single-elimination winners.}

\section{Conclusion}

In this paper, we investigate the behavior of a number of tournament solutions in large random tournaments under a general probabilistic model. We establish tight asymptotic bounds on the boundary of the probability range for which each tournament solution is unlikely to exclude any alternative. In particular, we illustrate a difference between the discriminative power of the top cycle and the uncovered set; this difference is not evident in previous studies that focused on more restricted models. Indeed, while both tournament solutions include all alternatives with high probability in the uniform random model, our results suggest that the uncovered set is in fact considerably more discriminative than the top cycle.

Our work leaves many interesting open questions for future study. A natural next step would be to investigate the asymptotic behavior of other tournament solutions that have been previously studied in the uniform random model---including the Banks set \citep{Fey08a}, the minimal covering set \citep{ScFe11a}, and the bipartisan set \citep{FiRy95a}---using our general probabilistic model. For instance, it is conceivable that the approach used by \cite{Fey08a} to show that the Banks set almost never rules out any alternative in the uniform random model can be extended to establish an analogous statement when each edge probability is drawn from some constant range. It is not clear, however, whether the approach would still work if we allow the range to depend on the number of alternatives in the tournament like we do in the current work.

From a broader point of view, we believe that an important direction is to apply our model to other tournament problems beyond those concerning tournament solutions, for example the problem of finding a dominating set of minimum size. It is well-known that a dominating set of size at most $\log_2(n+1)$ always exists and can be found using a simple greedy algorithm. While a dominating set can be as small as a singleton in tournaments that admit a Condorcet winner, \cite{ScFe11a} showed that for uniform random tournaments, a dominating set of logarithmic size is the best that one can hope for. More precisely, these authors showed that given any constant $0<c<1$, the smallest dominating set of a tournament chosen uniformly at random contains at least $c\log_2n$ alternatives with high probability. Establishing a similar result in our general probabilistic model is an intriguing technical challenge that would allow us to better understand the behavior of such structures in the real world.

\section*{Acknowledgments}

This material is based upon work supported by the Deutsche Forschungsgemeinschaft under grant {BR~2312/11-1}, by the European Research Council (ERC) under grant number 639945 (ACCORD), and by a Stanford Graduate Fellowship. The majority of the work was done while the second author was a PhD student at Stanford University and visiting the Technical University of Munich. A preliminary version of the paper appeared in Proceedings of the 14th Conference on Web and Internet Economics. The authors thank Felix Brandt, Pasin Manurangsi, and Fedor Petrov for helpful discussions.

\bibliography{../pamas/abb,../pamas/group}

\appendix

\section{Appendix}
\subsection{Proof of Lemma~\ref{lem:majorization-sum}}

Before we prove the lemma, we first give a characterization of when one vector majorizes another. To this end, we define the notion of an equalizing move, which involves taking two components of a vector and bringing them ``closer together''.

\begin{definition}
Given a vector $\textbf{x}$, an \emph{equalizing move} on $\textbf{x}$ takes two components $x_i>x_j$ and replaces them by $x_i-1$ and $x_j+1$, respectively.
\end{definition}

\begin{lemma}
\label{lem:majorization-equivalence}
For any vector $\textbf{x}$, a vector $\textbf{y}$ can be obtained from $\textbf{x}$ by a finite number of equalizing moves if and only if $\textbf{x}\succ\textbf{y}$.
\end{lemma}

\begin{proof}
The direction from left to right follows from the observation that an equalizing move never decreases the sum of the $j$ highest components of the vector for any $j=1,2,\dots,n-1$, and leaves the sum of all components invariant. 

For the converse direction, we proceed by induction on $n$, the number of components of the vectors. The base case $n=1$ holds trivially. Assume that the statement holds when there are at most $n-1$ components. To prove the statement when there are $n$ components, assume for contradiction that it does not hold for some pairs $\textbf{x},\textbf{y}$, i.e., $\textbf{x}\succ\textbf{y}$ but $\textbf{y}$ cannot be obtained from $\textbf{x}$ by a finite number of equalizing moves. For each pair assume without loss of generality that $x_1\geq\dots\geq x_n$ and $y_1\geq\dots\geq y_n$, and consider the pairs that minimize the difference $x_1-y_1$ among all such pairs; this difference is guaranteed to be nonnegative by condition (i) of Definition~\ref{def:majorization}. Among all of the pairs under consideration, take one that minimizes the largest index $k$ such that $x_k=x_1$. Let $l$ be the smallest index such that $\sum_{i=1}^l x_i=\sum_{i=1}^l y_i$; the existence of $l$ is guaranteed by condition (ii) of Definition~\ref{def:majorization}. If $l<n$, we can apply the induction hypothesis on the first $l$ components and the last $n-l$ components separately to obtain $\textbf{y}$ from $\textbf{x}$ using a finite number of equalizing moves, which would be a contradiction. Hence we may assume that $l=n$. In particular, $\sum_{i=1}^j x_i\geq \sum_{i=1}^j y_i+1$ for all $j=1,2,\dots,n-1$.

Let $m$ be the smallest index such that $x_m=x_n$. If $x_k=x_m$ or $x_k=x_m+1$, then the only vector with nonincreasing components that is majorized by $\textbf{x}$ is $\textbf{x}$ itself, a contradiction. So $x_k\geq x_m+2$, and we may replace $(x_k,x_m)$ by $(x_k-1,x_m+1)$ in an equalizing move. Let $\textbf{x}'=(x_1',x_2',\dots,x_n')$ be the vector resulting from this move, i.e., $x_k'=x_k-1$, $x_m'=x_m+1$, and $x_i'=x_i$ for all $i\not\in\{k,m\}$. By definition of $k$ and $m$, we have $x_1'\geq\dots\geq x_n'$. Moreover, $\sum_{i=1}^n x_i'=\sum_{i=1}^n y_i$, and for any $j=1,2,\dots,n-1$ we have $\sum_{i=1}^j x_i'\geq \sum_{i=1}^j x_i-1\geq \sum_{i=1}^j y_i$. This means that $\textbf{x}'\succ\textbf{y}$. If $k=1$, we have $x_1'-y_1<x_1-y_1$, which means that we can make a sequence of equalizing moves on $\textbf{x}'$ to obtain $\textbf{y}$, a contradiction. Else, if $k>1$, then we have $x_1'-y_1=x_1-y_1$ and $x_k'<x_1'$, so we can again make a sequence of equalizing moves on $\textbf{x}'$ to obtain $\textbf{y}$ and arrive at a contradiction, completing our proof.
\end{proof}

We now proceed to the proof of Lemma~\ref{lem:majorization-sum}. Suppose that $\textbf{x}\succ\textbf{y}$, and fix $k\in\{1,2,\dots,n\}$. By Lemma~\ref{lem:majorization-equivalence}, there exists a sequence of equalizing moves that takes $\textbf{x}$ to $\textbf{y}$. It suffices to show that if an equalizing move takes $\textbf{x}$ to $\textbf{x}'$, then there is a corresponding sequence of equalizing moves that takes $\textbf{x}^{(k)}$ to $\textbf{x}'^{(k)}$. Indeed, if this is true, then the sequence of equalizing moves that takes $\textbf{x}$ to $\textbf{y}$ gives rise to a corresponding sequence of equalizing moves that takes $\textbf{x}^{(k)}$ to $\textbf{y}^{(k)}$. By Lemma~\ref{lem:majorization-equivalence} again, this will imply that $\textbf{x}^{(k)}\succ \textbf{y}^{(k)}$.

Consider an equalizing move that takes $\textbf{x}$ to $\textbf{x}'$; assume that the move replaces the components $x_i>x_j$ by $x_i-1$ and $x_j+1$, respectively. Note that the only components that change between $\textbf{x}^{(k)}$ and $\textbf{x}'^{(k)}$ are the ones that contain exactly one of $x_i$ and $x_j$ in their sum. These components can be paired up in such a way that for each pair, one component contains $x_i$, the other component contains $x_j$, and both components contain exactly the same subset of the remaining $x_l$'s with $l\not\in\{i,j\}$. For each pair, replacing $x_i$ and $x_j$ by $x_i-1$ and $x_j+1$ corresponds to an equalizing move. It follows that there exists a sequence of equalizing moves that takes $\textbf{x}^{(k)}$ to $\textbf{x}'^{(k)}$, as claimed.

\subsection{Proof of Lemma~\ref{lem:tournament-majorize}}

Note that $(\deg_W(d_1),\deg_W(d_2),\dots,\deg_W(d_n))=(n-1,n-2,\dots,0)$. We verify that both conditions in Definition~\ref{def:majorization} are satisfied.

Fix $k\in\{1,2,\dots,n-1\}$, and assume without loss of generality that $\deg_U(b_1)\geq\dots\geq \deg_U(b_n)$. Let $B=\{b_1,b_2,\dots,b_n\}$ and $B'=\{b_1,b_2,\dots,b_k\}$. The number of edges from an alternative in $B'$ to another alternative in $B'$ is exactly $\binom{k}{2}$. On the other hand, the number of edges from an alternative in $B'$ to an alternative in $B\backslash B'$ is at most $k(n-k)$. It follows that
\begin{align*}
\deg_U(b_1)+\deg_U(b_2)+\dots+\deg_U(b_k)
&\leq \binom{k}{2} + k(n-k) \\
&= k\left(n-\frac{k+1}{2}\right) \\
&= (n-1)+(n-2)+\dots+(n-k),
\end{align*}
so condition (i) is satisfied.

Finally, observe that 
\[
\deg_U(b_1)+\deg_U(b_2)+\dots+\deg_U(b_n)
=\binom{n}{2} 
=(n-1)+(n-2)+\dots+0,
\]
so condition (ii) is also satisfied.

\end{document}